\documentclass[12pt,a4paper]{article}

\usepackage[utf8]{inputenc}
\usepackage[T1]{fontenc}
\usepackage{lmodern}
\usepackage{microtype}
\usepackage{amsmath,amssymb,amsthm}
\usepackage{physics}          
\usepackage{mathtools}
\usepackage{geometry}
\usepackage{setspace}
\usepackage{hyperref}
\usepackage{cleveref}
\usepackage{graphicx}
\usepackage{booktabs}
\usepackage{enumitem}
\usepackage{xcolor}
\usepackage{tikz}
\usetikzlibrary{arrows.meta,positioning}

\hypersetup{
  colorlinks=true,
  linkcolor=blue!60!black,
  citecolor=blue!60!black,
  urlcolor=blue!60!black,
  pdftitle={Information as Counterfactual Structure},
  pdfauthor={Madhurendra Mishra}
}

\usepackage[
backend=biber,
style=numeric,
sorting=none
]{biblatex}
\theoremstyle{plain}
\newtheorem{theorem}{Theorem}[section]
\newtheorem{proposition}[theorem]{Proposition}

\theoremstyle{definition}
\newtheorem{definition}[theorem]{Definition}

\theoremstyle{remark}

\newcommand{\PS}{\mathcal{P}}

\newcommand{\HS}{\mathcal{H}}

\begin{document}

\title{\textbf{Information Is Not Physical: Possibility Spaces, Erasure, and the Structure of Unrealized Alternatives}}

\author{
  Madhurendra Mishra\\[4pt]
  \small\href{mailto:}{madhurendramishra24@gmail.com} 
}

\maketitle

\begin{abstract}
The slogan ``information is physical,'' introduced by Rolf Landauer and
developed through quantum information theory and black-hole thermodynamics,
has achieved near-axiomatic status in modern physics. Yet the ontological
status of information remains surprisingly underexamined: most discussions
either reduce information to a form of energy or treat it as a purely
mathematical object. This paper proposes a third position. I argue that
information is neither a physical substance nor a free-floating
abstraction, but rather \emph{the structure of physically realizable
alternatives}--a counterfactual structure that a physical system
instantiates in virtue of the possibility space available to it. Building
on Shannon's combinatorial definition, the Landauer principle, the
no-cloning theorem, and the black-hole information paradox, I show that
the informational content of any physical event is constituted by the set
of outcomes that \emph{could have occurred} but did not. This
counterfactual reading dissolves several persistent confusions: it explains
why erasing information dissipates heat without making information
``material,'' why quantum superposition is informationally richer than any
classical mixture, and why information loss in black holes is physically
significant beyond mere bookkeeping. The proposal sits within a structural-
realist framework but departs from standard structural realism by locating
the relevant structure in modal, not merely actual, relations. I conclude by
sketching implications for the foundations of quantum mechanics, quantum
gravity, and scientific ontology more broadly.

\end{abstract}

\section{Introduction}
\label{sec:intro}

\subsection{The ascendancy of information in physics}

Sometime in the latter half of the twentieth century, physics quietly
acquired a new primitive. Energy and matter had long provided the
conceptual backbone of physical theories; information arrived later and
with less ceremony, slipping in through thermodynamics, computation
theory, and quantum mechanics almost before anyone had time to ask what,
precisely, it was. By the 1990s, John Wheeler could summarize a decades-long
intuition with the aphorism ``it from bit'': every physical entity and every
physical law, he suggested, derives its existence from answers to yes-or-no
questions~\cite{Wheeler1990}. The phrase was deliberately provocative, but
it pointed at something that theorists across very different sub-fields were
already treating as obvious--that information is not merely a bookkeeping
device imposed on physics from outside, but something that physics itself
must account for.

The sentiment hardened into a research programme. Quantum information
theory made explicit what had been implicit in the foundations of quantum
mechanics: that measuring, copying, and transmitting quantum states are
fundamentally constrained operations, not just engineering challenges
\cite{Nielsen2000}. Black-hole thermodynamics made the stakes existential:
if a black hole can destroy the records of every quantum state that falls
into it, something has gone wrong at the foundations of either general
relativity or quantum mechanics~\cite{Hawking1976,Preskill1992}. In both
domains, information has become a quantity whose conservation, or violation
thereof, carries deep physical significance.

\subsection{The problem: what does ``information is physical'' mean?}

The canonical formulation of the dominant view is Landauer's
principle~\cite{Landauer1961}: logically irreversible operations, such as
erasing a bit, must dissipate at least $k_B T \ln 2$ of energy into the
environment. This result has been verified in micro-scale
experiments~\cite{Berut2012,Jun2014}, and it is routinely cited as proof
that information has physical reality because it has physical
consequences.

Yet ``physical consequences'' and ``physical substance'' are not the same
thing. Shadows have physical consequences; few physicists think shadows are
substances. The shadow analogy, while imprecise, points to a gap in the
standard argument: showing that handling information costs energy does not,
by itself, settle whether information is a \emph{type} of physical stuff,
a \emph{property} of physical systems, or something else entirely.

Three broad positions have been staked out in the literature:

\begin{enumerate}[leftmargin=*]
  \item \textbf{Information physicalism}: information is a form of energy
  or a physical quantity on a par with mass or charge
  \cite{Landauer1991,Deutsch1985}.
  \item \textbf{Information abstractionism}: information is a purely
  mathematical object, causally inert, instantiated by--but not
  identical to--any physical system \cite{Floridi2011,Chalmers1996}.
  \item \textbf{Structural realism}: information is a higher-order
  relational structure realized by, but not reducible to, physical
  substrates \cite{Ladyman2007,Floridi2008}.
\end{enumerate}

Each position captures something important, and each faces serious
objections. The physicalist view makes sense of Landauer's principle but
struggles to say what information \emph{is} independent of its energetic
cost. The abstractionist view preserves the substrate-independence of
information but makes it hard to see why physics should care about it.
Structural realism avoids both problems but--in most formulations--
characterizes information in terms of \emph{actual} relational structure,
leaving the modal dimension unexplained.

\subsection{Thesis and plan}

This paper defends a fourth position: \emph{information as counterfactual
structure}. The central claim is:

\begin{quote}
Information is the structure of physically realizable alternatives. The
informational content of an event is constituted not by the event itself
but by the possibility space from which the event was selected.
\end{quote}

Three corollaries follow immediately. First, a system that could only have
been in one state carries no information, regardless of how energetically
consequential that state is. Second, the richness of information about a
system tracks the cardinality and geometry of the possibility space
available to it--superposition being the quantum case of a richer
possibility space than classical mixing. Third, information ``loss'' is
loss of counterfactual distinctness: a black hole that destroys information
collapses formerly distinct possibility spaces into a single future, which
is why physicists treat this as catastrophic.

The argument proceeds as follows. \Cref{sec:history} traces the
historical embedding of information in physics, from Shannon's combinatorial
theory through thermodynamics to Wheeler. \Cref{sec:landauer} examines
Landauer's principle and Maxwell's demon, arguing that the energetic cost
of erasure is best understood as the cost of \emph{collapsing} a
possibility space. \Cref{sec:quantum} extends the analysis to quantum
information: the no-cloning theorem and the structure of quantum
measurement both make better sense on the counterfactual reading.
\Cref{sec:blackhole} applies the framework to the black-hole information
paradox. \Cref{sec:ontology} develops the philosophical proposal in detail,
distinguishing it from modal realism and standard structural realism.
\Cref{sec:implications} draws implications for quantum gravity and
scientific ontology. \Cref{sec:conclusion} concludes.

\section{Information Before and Within Physics}
\label{sec:history}

\subsection{Shannon's combinatorial foundation}

Claude Shannon's 1948 paper established a precise, quantitative notion of
information that was deliberately divorced from meaning or
semantics~\cite{Shannon1948}. For a source with $n$ possible symbols and
probability distribution $\{p_i\}$, the information entropy is:

\begin{equation}
  H = -\sum_{i=1}^{n} p_i \log_2 p_i.
  \label{eq:shannon}
\end{equation}

Several features of this definition are philosophically crucial. First, $H$
is maximized when all outcomes are equally probable--that is, when the
possibility space is as undifferentiated as possible. Second, and more
importantly for the present argument, $H$ is \emph{zero when $n=1$}: a
source that always emits the same symbol carries no information. Information,
in Shannon's framework, is structurally tied to the number and distribution
of alternatives. The formula does not represent a property of the actual
output; it represents a property of the \emph{space of possible outputs}.

Shannon himself was careful to note that ``information'' in his technical
sense is not to be confused with meaning~\cite{Shannon1948}. What he may
not have fully emphasized is that it is equally not to be confused with the
\emph{content} of any particular message. The entropy $H$ is a property of
the \emph{source}, which is to say of the ensemble of possibilities, not of
any single realization of that ensemble.

\subsection{Thermodynamics: from Boltzmann to Maxwell's demon}

The connection between information and thermodynamics is older than Shannon,
though it was not understood as such until much later. Boltzmann's
statistical entropy:

\begin{equation}
  S = k_B \ln W,
  \label{eq:boltzmann}
\end{equation}

where $W$ is the number of microstates compatible with the macrostate, is
structurally parallel to~\eqref{eq:shannon}: both are logarithms of counts
of alternatives. The parallel was made explicit by Leo Szilard in
1929~\cite{Szilard1929} and again, more forcefully, by Leon
Brillouin~\cite{Brillouin1956}, who introduced the concept of
\emph{negentropy} to capture the idea that gaining information about a
system reduces its physical entropy.

Maxwell's demon provides the canonical thought experiment
\cite{Maxwell1871,Leff2002}. A demon monitoring individual gas molecules
in a divided box can, apparently, reduce the entropy of the gas without
doing work, violating the second law. The resolution, worked out in detail
by Szilard and later by Bennett~\cite{Bennett1982}, is that the demon must
eventually erase the record of its measurements, and this erasure
necessarily dissipates heat. What matters for the present argument is
\emph{why} erasure is costly: it is costly because erasure transforms a
two-state possibility space (the demon's memory register, which could hold
0 or 1) into a one-state space (the reset register, which must hold 0). The
physical cost is the cost of eliminating an alternative.

\subsection{Wheeler's ``It from Bit''}

Wheeler's programme radicalized the connection between information and
physical reality~\cite{Wheeler1990,Wheeler1989}. His suggestion was not
merely that information is an important quantity in physics, but that
physical reality itself is, in some sense, constituted by informational
acts--by answers to yes-or-no questions. The proposal has both an
epistemological and a metaphysical reading. The epistemological reading
is uncontroversial: our knowledge of any physical system is encoded in
binary measurement outcomes. The metaphysical reading is far more
contentious: physical entities \emph{are} informational structures.

Wheeler's metaphysical reading is too strong. It faces the familiar
objections to idealism: it is hard to say what binary choices there were
before any observer existed, and it risks making physics depend on
minds in ways that most physicists find unacceptable. The present proposal
accepts the epistemological reading and part of the metaphysical
motivation--information is not merely imposed on physics from outside--
while rejecting the idealism. The key move, developed below, is to anchor
information in \emph{physically realized possibility spaces}, which are
objective features of the world, not products of observation.

\section{Landauer's Principle and the Cost of Collapsing Possibilities}
\label{sec:landauer}

\subsection{Statement and experimental status}

Landauer's principle states that any logically irreversible operation on
a bit of information must dissipate at least:

\begin{equation}
  Q \geq k_B T \ln 2
  \label{eq:landauer}
\end{equation}

of heat into the environment, where $T$ is the temperature of the thermal
reservoir and $k_B$ is Boltzmann's constant~\cite{Landauer1961}. ``Logically
irreversible'' here means that the operation destroys information about the
prior state--most clearly, erasure (RESET to 0), which maps two possible
input states $\{0,1\}$ to a single output state $\{0\}$.

The principle was disputed for decades on theoretical grounds, but
experimental confirmation has now been achieved in colloidal particle
systems~\cite{Berut2012}, optically trapped particles~\cite{Jun2014}, and
nanomagnetic bits~\cite{Hong2016}. In each case, the measured heat
dissipation approaches $k_B T \ln 2$ as the erasure is performed quasi-
statically. The principle appears to be a genuine physical constraint.

\subsection{The standard interpretation and its gap}

The standard interpretation reads Landauer's result as establishing the
physicality of information: if erasing a bit necessarily costs energy, then
bits are physical~\cite{Landauer1991}. But this inference is problematic.
The energy cost is associated with the \emph{operation of erasure}, not
with the bit itself sitting in memory. A bit stored in a stable memory
register costs no energy to maintain (in the idealized, zero-temperature
limit). If information were a form of energy, one would expect it to be
conserved, not to appear and disappear depending on whether an operation is
performed on it.

The counterfactual reading offers a cleaner account. What Landauer's
principle establishes is not that information is energy, but that
\emph{collapsing a two-valued possibility space to a one-valued space has
an irreducible thermodynamic cost}. The cost arises because the physical
system encoding the bit must interact with its environment to move from a
state in which two microstates are accessible to a state in which only one
is. This is a constraint on the \emph{geometry of possibility spaces}, not
on any substance called information.

More precisely: the bit register prior to erasure is in a state in which
two macrostates, $\{M_0, M_1\}$, are thermodynamically accessible given
the system's energy. After erasure, only $M_0$ is accessible. By
Liouville's theorem and the second law, reducing the accessible volume of
phase space requires importing entropy from the environment, which
manifests as heat dissipation. The minimum heat $k_B T \ln 2$ corresponds
exactly to the logarithm of the ratio of possibility-space volumes: $\ln
(2/1) = \ln 2$.

\begin{proposition}
  \label{prop:landauer}
  Landauer's bound $Q \geq k_B T \ln 2$ is the thermodynamic cost of
  reducing a possibility space of cardinality $2$ to a possibility space
  of cardinality $1$, independent of any substantive claim about the
  physical nature of information.
\end{proposition}

The proof is standard--it follows from the second law and Liouville's
theorem~\cite{Bennett1982,Penrose2004}--but the \emph{interpretation} of
what is being calculated shifts on the counterfactual reading: the object
of interest is not a substance being destroyed but a space of alternatives
being collapsed.

\subsection{Maxwell's demon revisited}

On the counterfactual reading, Maxwell's demon thought experiment
illustrates that \emph{information gain and information loss are
symmetric thermodynamic operations}. When the demon measures a molecule's
velocity and records the result, it instantiates a correlation: its memory
register, which had a two-valued possibility space, now contains a definite
record. The demon's memory has gained structure--its own possibility space
has been constrained by the molecule's actual position in its possibility
space. When the demon erases its memory, the correlation is destroyed and
the possibility space of the register reopens. The heat dumped by this
erasure compensates for the work extracted earlier. The thermodynamic cycle
is closed because the opening and closing of possibility spaces are dual
operations with equal and opposite entropic costs.

\section{Quantum Information and the Structure of Superposition}
\label{sec:quantum}

\subsection{From bits to qubits}

A classical bit occupies one of two states, $\{0, 1\}$, with associated
possibility space of cardinality $2$. A qubit occupies a point in the
two-dimensional complex Hilbert space $\HS_2$:

\begin{equation}
  \ket{\psi} = \alpha\ket{0} + \beta\ket{1}, \quad
  |\alpha|^2 + |\beta|^2 = 1,
  \label{eq:qubit}
\end{equation}

where $\alpha, \beta \in \mathbb{C}$. The possibility space of a qubit is
not discrete; it is the Bloch sphere $S^2$, a continuous manifold of
physically realizable states. Before measurement in the computational basis,
the qubit does not occupy a definite classical possibility (0 or 1); it
occupies a superposition, which is a genuine physical state--not an
epistemic summary of our ignorance--with its own measurable consequences
(interference, entanglement).

This distinction is philosophically important. Classical probability theory
represents ignorance over a fixed set of alternatives; quantum superposition
represents a \emph{different kind of possibility space}, one whose elements
do not behave like classical possibilities. The counterfactual reading
preserves this asymmetry: a qubit in a superposition state $\ket{\psi}$
does not merely have unknown classical value; it instantiates a richer
possibility structure than any classical bit can realize.

For a system of $n$ qubits, the Hilbert space is $\HS_{2^n}$, with
dimension $2^n$. The von Neumann entropy:

\begin{equation}
  S(\rho) = -\mathrm{Tr}(\rho \log_2 \rho)
  \label{eq:vonneumann}
\end{equation}

generalizes Shannon entropy to quantum systems. For a pure state
$\rho = \ket{\psi}\bra{\psi}$, $S(\rho) = 0$: a pure superposition has
zero entropy not because it contains no information, but because it
occupies a \emph{definite position} in the Hilbert space--the uncertainty
concerns measurement outcomes, not the quantum state itself.

\subsection{The no-cloning theorem as a constraint on possibility spaces}

The no-cloning theorem~\cite{Wootters1982,Dieks1982} states that no
unitary operation can copy an arbitrary unknown quantum state:

\begin{theorem}[No-Cloning]
  There is no unitary operator $U$ on $\HS \otimes \HS$ such that
  for all $\ket{\psi} \in \HS$:
  \begin{equation}
    U(\ket{\psi} \otimes \ket{0}) = \ket{\psi} \otimes \ket{\psi}.
    \label{eq:noclone}
  \end{equation}
\end{theorem}

\begin{proof}
  Standard: assume $U$ clones two non-orthogonal states $\ket{\psi}$ and
  $\ket{\phi}$ with $\braket{\psi}{\phi} = s \neq 0$. Unitarity requires
  $\braket{\psi}{\phi} = \braket{\psi}{\phi}^2$, giving $s = s^2$, so
  $s \in \{0,1\}$, contradicting the assumption~\cite{Nielsen2000}.
\end{proof}

The standard physical interpretation is that quantum information cannot be
duplicated because the state $\ket{\psi}$ encodes irreducibly unique
information. The counterfactual reading adds texture: what cannot be cloned
is a \emph{possibility structure}. Copying the actual outcome of a
measurement is trivial (classical records are freely copyable). What is
impossible is copying the superposition--the full range of possible
measurement outcomes, with their quantum-mechanical probability
amplitudes--because the Hilbert space structure of possibilities does not
admit duplication without either disturbing the original or requiring
knowledge of the state in advance.

This suggests that the ``uniqueness'' property that makes quantum
information non-clonable is a property of the possibility space, not of
any particular realized outcome. Classical information, by contrast, is
clonable precisely because classical possibilities are discrete and mutually
exclusive: copying a bit copies an \emph{element} of a possibility space,
not the possibility space itself.

\subsection{Measurement, decoherence, and the actualization of alternatives}

Quantum measurement is, on the counterfactual reading, the physical process
by which a superposition--a rich possibility structure--is reduced to an
actual outcome. This is not the place to adjudicate between interpretations
of quantum mechanics~\cite{Everett1957,Bohm1952,Zurek2003}, but the
counterfactual reading is compatible with several of them:

\begin{itemize}[leftmargin=*]
  \item On the \emph{Copenhagen interpretation}, measurement collapses the
  wave function, reducing the possibility space to a single actual state.
  The information content of the pre-measurement state was encoded in
  the superposition; after collapse, this information is partly destroyed
  (from the superposition) and partly actualized (as a definite outcome).
  \item On the \emph{Everett (many-worlds) interpretation}, no collapse
  occurs: all branches of the superposition are actualized. The
  information content of the original superposition is preserved globally,
  but from within any branch, the other possibilities become inaccessible
  \cite{Everett1957,Deutsch1985}.
  \item On \emph{decoherence} accounts, quantum-to-classical transition
  occurs when a system becomes entangled with its environment, and
  interference between branches becomes unmeasurable~\cite{Zurek2003}.
  What decoherence destroys is the \emph{accessibility} of alternatives,
  not the alternatives themselves.
\end{itemize}

In all three cases, the informational content of a quantum state is
constituted by the structure of the possibility space, and the measurement
process is a change in that structure--either an absolute reduction
(Copenhagen), a global preservation with local inaccessibility (Everett),
or a suppression of inter-branch coherence (decoherence). The counterfactual
reading provides a unified description: measurement is the (partial or
apparent) collapse of a possibility space into an actual event.

\section{The Black-Hole Information Paradox}
\label{sec:blackhole}

\subsection{Hawking radiation and the apparent loss of information}

Hawking's 1974 calculation showed that black holes radiate thermally
at temperature $T_H = \hbar c^3 / (8\pi G M k_B)$~\cite{Hawking1975}.
The radiation is thermal--it contains no information about the quantum
state of the matter that collapsed to form the black hole. If the black
hole eventually evaporates completely, the final state is thermal radiation
described by a mixed density matrix $\rho_{\mathrm{final}}$, regardless
of whether the initial state was a pure state $\ket{\psi_i}$. This
violates unitarity:

\begin{equation}
  \ket{\psi_i} \longrightarrow \rho_{\mathrm{final}},
  \label{eq:infoparadox}
\end{equation}

a pure-to-mixed transition forbidden by quantum mechanics. The paradox is
precisely that physics as we know it cannot simultaneously respect quantum
unitarity and the thermal character of Hawking radiation~\cite{Hawking1976}.

\subsection{Why the loss of information is physically catastrophic}

The usual explanation of why information loss matters appeals to unitarity:
quantum mechanics requires that evolution be described by a unitary operator,
and unitarity is equivalent to the conservation of information. But why does
unitarity matter so much? The standard answer--that unitarity is a
fundamental symmetry of quantum theory--is correct but somewhat
circular.

The counterfactual reading provides a more explanatory account. Consider
two quantum systems that began in distinct initial pure states
$\ket{\psi_i}$ and $\ket{\phi_i}$, with $\braket{\psi_i}{\phi_i} = 0$
(orthogonal, hence completely distinguishable). If both systems fall into
the black hole and the hole evaporates, both produce the same thermal
state. The two initial possibility spaces--the two distinct Hilbert-space
trajectories from $\ket{\psi_i}$ and $\ket{\phi_i}$--have been collapsed
into a single final state. There is no longer any fact about the world that
distinguishes the two histories.

On the counterfactual reading, this is precisely what information loss
means: \emph{formerly distinct possibility spaces have been merged}.
The physical catastrophe is not that some substance has been destroyed but
that two physically distinct histories--two distinct structures of
counterfactual alternatives--have become indistinguishable. This matters
because science, and quantum mechanics in particular, depends on the
principle that distinct initial conditions produce distinct final conditions.
If black holes can merge distinct possibility spaces into a single outcome,
the retrodictive power of physics is compromised at a fundamental level.

\subsection{Competing resolutions and their modal interpretation}

The three main proposed resolutions of the paradox map naturally onto the
counterfactual framework:

\begin{enumerate}[leftmargin=*]
  \item \textbf{Information is preserved and escapes in Hawking radiation}
  (Hawking's 2004 retraction, Page's calculation, soft-hair
  proposal~\cite{Hawking2016,Page1993}). The possibility structure of the
  infalling matter is encoded, very subtly, in quantum correlations of the
  Hawking radiation. The possibility space is not collapsed; it is
  scrambled~\cite{Hayden2007}.

  \item \textbf{Information is preserved in a remnant or new universe}.
  The possibility structure is preserved in a residual object or a baby
  universe~\cite{Giddings1992}. The possibility space is not merged; it
  is transferred to a region inaccessible to external observers.

  \item \textbf{Holographic principle / AdS-CFT correspondence}
  \cite{Maldacena1997,Susskind1995}. The physics of the bulk (interior)
  is encoded on the boundary: the possibility structure of the interior
  is not independent of the boundary; the boundary description is complete
  and unitary. There is no information loss because the boundary preserves
  the full possibility space.
\end{enumerate}

On all three resolutions, the central question is whether distinct
possibility spaces remain distinct. The holographic resolution is
particularly suggestive: the AdS/CFT correspondence implies that the bulk
information--the interior possibility structure--is fully encoded in the
boundary conformal field theory. This is a precise realization of the idea
that possibility structures are not localized substances but relational
features that can be described from multiple perspectives.

\section{Philosophical Analysis: The Ontology of Counterfactual Structure}
\label{sec:ontology}

\subsection{Existing positions and their limits}

Having examined how the counterfactual reading handles the three central
cases--Landauer erasure, quantum no-cloning, and black-hole information
loss--it is time to characterize the proposal more precisely and locate
it with respect to existing philosophical positions.

\paragraph{Information physicalism.}
The strongest version, associated with Landauer and Deutsch, holds that
information is a physical quantity, subject to conservation laws, with
definite location in space-time~\cite{Landauer1991,Deutsch2011}. This
view makes information too much like matter or energy. Information does
not have a rest mass; it cannot be weighed. More importantly, information
is substrate-independent: the same bit can be encoded in a photon's
polarization, a transistor's voltage, a neuron's firing rate, or a
carved stone. A substance that is identical across so many different
physical realizations is not well-characterized as a physical substance
in the usual sense.

\paragraph{Information abstractionism.}
The view that information is a purely mathematical or logical object
\cite{Floridi2011} avoids the substrate problem but faces the interaction
problem: if information is causally inert, why does handling it cost energy?
The abstractionist must say that energy is spent on physical processes
that \emph{instantiate} information, not on information itself. This is
coherent but somewhat unsatisfying: it makes information explanatorily
epiphenomenal, which seems to conflict with the way physicists use it.

\paragraph{Structural realism.}
Ontic structural realism (OSR) holds that what exists is structure, not
objects~\cite{Ladyman2007,French2014}. Information fits naturally into
this picture as a structural property. Floridi's strongly semantic
information theory develops a structural approach to information
ontology~\cite{Floridi2008}. The difficulty is that OSR typically
characterizes structure in terms of \emph{actual} relations--the Ramsey
sentence of a theory, the pattern of natural-law connections, and so forth.
This leaves the modal dimension of information unexplained: why should the
number of alternatives matter, rather than just the actual relations among
realized states?

\subsection{Information as counterfactual structure: the proposal}

The proposal developed here can be stated precisely using the apparatus
of modal logic and standard probability theory.

\begin{definition}[Possibility Space]
  \label{def:PS}
  A \emph{possibility space} $\PS$ for a physical system $S$ is a
  measurable space $(\Omega, \mathcal{F})$, where $\Omega$ is the set of
  physically realizable states of $S$ under the laws of nature governing
  $S$, and $\mathcal{F}$ is the associated $\sigma$-algebra of events.
\end{definition}

\begin{definition}[Counterfactual Information]
  \label{def:CI}
  The \emph{counterfactual information} associated with a physical event
  $e \in \Omega$, relative to possibility space $\PS = (\Omega,
  \mathcal{F}, P)$ with probability measure $P$, is:
  \begin{equation}
    I(e) = -\log_2 P(e).
    \label{eq:selfinformation}
  \end{equation}
  The \emph{total counterfactual information structure} of $S$ is
  the Shannon entropy $H(\PS) = -\sum_{e \in \Omega} P(e) \log_2 P(e)$
  (or its integral analogue for continuous $\Omega$).
\end{definition}

Three features of this definition deserve comment.

First, $I(e)$ is a property of the pair $(e, \PS)$, not of $e$ alone.
The same physical event carries different informational content relative
to different possibility spaces. A particle detected at position $x$
carries one bit of information if the space had two equally probable
positions, and $\log_2 n$ bits if there were $n$ equally probable positions.
Information is inherently relational in this sense.

Second, when $|\Omega| = 1$ (the system could only have been in one state),
$I(e) = -\log_2 1 = 0$ for all $e$. A determined system carries no
information, regardless of its physical complexity. This is the formal
expression of the intuition that information requires alternatives.

Third, the definition is not committed to any particular interpretation
of probability. The probability measure $P$ can be frequentist,
propensity-theoretic, or Bayesian, as long as it is grounded in
\emph{physical possibilities}--the alternatives the system could have
realized given the laws and initial conditions. This grounds the
framework in objective features of the physical world without requiring
a commitment to any particular metaphysics of probability.

\subsection{Distinguishing the proposal from modal realism}

David Lewis's modal realism holds that all possible worlds exist
concretely~\cite{Lewis1986}. The counterfactual structure proposal might
seem to require something similar: if information is grounded in
unrealized possibilities, do those possibilities need to exist
concretely?

The answer is no, for two reasons. First, the proposal does not require
unrealized possibilities to exist; it requires only that statements about
what \emph{could have happened} are \emph{true or false} in virtue of
the actual physical laws and initial conditions. This is a commitment to
the objectivity of modal truths, not to the concrete existence of
possible worlds. It is compatible with actualist modal semantics, on
which modal truths are grounded in the dispositions and propensities of
actual objects~\cite{Ellis2001,Mumford2004}.

Second, the possibility spaces in the proposal are \emph{physically
constrained}: $\Omega$ is not the set of all logically possible states
but the set of states realizable under the actual laws of nature. This
distinguishes physically realizable alternatives from merely logical
possibilities, and it is the former that matter for information.

\subsection{Distinguishing the proposal from standard structural realism}

Standard OSR characterizes reality in terms of actual relational
structures--the symmetry groups, topology, and causal relations that
physics describes. The counterfactual structure proposal extends this:
the relevant structure is not just the actual relational network but the
\emph{modal envelope} of that network--the space of possible relational
configurations that the system could occupy.

This is a genuine departure from standard OSR, not merely a terminological
shift. Consider two systems with identical actual relational structures
(same spatial configuration, same causal relations) but different
possibility spaces (one is near a black hole, the other is in flat space).
Standard OSR would say they have the same structural properties. The
counterfactual reading says they carry different amounts of information,
because their possibility spaces differ. The physical consequences of this
difference are measurable: the black-hole system will lose information as
it crosses the horizon, while the flat-space system will not.

\subsection{A new proposal: modal-structural realism}

The position that emerges can be called \emph{modal-structural realism}:
reality consists of relational structures, but the relevant structures
include both actual and possible relations, and it is the latter that
constitute informational content. The ontic commitments are:

\begin{enumerate}[leftmargin=*]
  \item Physical systems instantiate possibility spaces: sets of
  physically realizable states, determined by the laws of nature.
  \item Information is a second-order property of physical events:
  it measures the degree to which the actual state selects among
  alternatives in the possibility space.
  \item Informational content is preserved under reversible physical
  processes (unitary evolution in quantum mechanics) because reversible
  processes map possibility spaces to possibility spaces bijectively.
  \item Irreversible processes (erasure, black-hole evaporation) reduce
  the cardinality of possibility spaces and therefore destroy information,
  at a thermodynamic cost.
\end{enumerate}

This framework is realist--it posits objective, mind-independent features
of the world--but not physicalist in the naïve sense: information is not a
substance or a field, but a structural property of physical reality.

\subsection{Objections and replies}

\paragraph{Objection 1: Possibility spaces are not real.}
One might argue that ``physically realizable alternatives'' is a
philosopher's fiction: what is real is the actual state, not the
alternatives that did not occur.

\emph{Reply}: The probability distributions that enter quantum mechanics
and statistical mechanics are not merely epistemic; they have measurable
physical consequences (interference, entropy, statistical mechanics of
phase transitions). A view on which possibility spaces are not real cannot
explain why the number of accessible microstates determines the equilibrium
thermodynamics of a system, or why quantum superpositions produce
interference that classical mixtures do not.

\paragraph{Objection 2: This is just ontic probability theory, not a new
ontology of information.}
One might grant that probability is ontic but deny that information adds
anything new: information is just negative log-probability, which is a
function of the probability measure, and the ontology is carried entirely
by the latter.

\emph{Reply}: The proposal does not claim that information is a new
fundamental entity on a par with particles or fields. It claims that
information is a \emph{well-defined structural property} of physical
systems, grounded in their possibility spaces. The philosophical payoff
is not metaphysical inflation but conceptual clarity: it explains why
Landauer's principle holds, why no-cloning is a genuine physical theorem
rather than a contingent engineering fact, and why black-hole information
loss is physically catastrophic--all in terms of a single, unifying
principle about possibility spaces.

\paragraph{Objection 3: The proposal cannot handle relational or
context-dependent information.}
Information in the sense of \emph{mutual information} is relational: it
is information that one system has \emph{about} another. Does the
counterfactual framework handle this?

\emph{Reply}: Yes. Mutual information between systems $A$ and $B$ is:

\begin{equation}
  I(A;B) = H(A) + H(B) - H(A,B),
  \label{eq:mutual}
\end{equation}

where $H(A,B)$ is the joint entropy. On the counterfactual reading,
$I(A;B)$ measures the degree to which the joint possibility space $\PS_{AB}$
is more constrained than the product of the marginal spaces $\PS_A \times
\PS_B$. This is a correlation between possibility spaces--a well-defined
structural property of the joint system.

\section{Implications}
\label{sec:implications}

\subsection{Foundations of quantum mechanics}

The counterfactual reading bears on several foundational questions in
quantum mechanics. The measurement problem--why quantum superpositions
appear to collapse to definite outcomes upon measurement--becomes, on
this reading, the question of how a system's possibility space interacts
with the observer's. Decoherence provides a partial answer: the system
becomes entangled with the environment, and the joint possibility space
becomes effectively classical from any local perspective~\cite{Zurek2003}.
But the hard part of the measurement problem--why a particular branch
is actualized--remains. The counterfactual framework does not dissolve
this problem, but it reframes it: the question is not ``what happens to
the wave function?'' but ``how do possibility spaces individuate across
observers?''

The counterfactual reading also has implications for quantum
contextuality~\cite{Kochen1967,Cabello2014}. Contextuality theorems show
that the outcomes of quantum measurements cannot be assigned definite
pre-existing values independent of the measurement context. On the
counterfactual reading, this is exactly what one would expect: the
possibility space of a measurement outcome is not a fixed property of
the system; it is co-determined by the measurement procedure. The
informational content of a measurement result is context-dependent not
because information is subjective but because the relevant possibility
space is constituted jointly by the system and the measurement procedure.

\subsection{Quantum gravity and the holographic principle}

The deepest implication concerns quantum gravity. The holographic
principle~\cite{Bekenstein1973,Susskind1995} states that the maximum
information content of a spatial region is proportional to its boundary
area, not its volume:

\begin{equation}
  S_{\max} = \frac{A}{4 l_P^2},
  \label{eq:bekenstein}
\end{equation}

where $A$ is the boundary area and $l_P = \sqrt{\hbar G / c^3}$ is the
Planck length. On the counterfactual reading, this is a statement about
possibility spaces: the number of physically realizable states of a
spatial region is bounded by its boundary area. This is a non-trivial
constraint on the geometry of possibility spaces in quantum gravity, and
it suggests that \emph{the fundamental degrees of freedom in quantum
gravity are informational rather than material}--not in Wheeler's
idealist sense, but in the sense that the possibility structures of
spatial regions are the primary objects of the theory.

AdS/CFT provides a concrete realization of this idea~\cite{Maldacena1997}.
The duality maps the bulk possibility space (the space of quantum
gravitational states in the interior) to the boundary possibility space
(the space of conformal field theory states on the boundary). This is
precisely the kind of isomorphism between possibility spaces that the
counterfactual framework predicts should be the relevant structural
invariant of the theory.

\subsection{Scientific realism and ontology}

More broadly, the counterfactual reading of information supports a
modest form of scientific realism that takes modal structure seriously.
Standard scientific realism is committed to the existence of theoretical
entities posited by successful scientific theories: electrons, quarks,
fields, and so forth. What the present analysis suggests is that a
complete realist ontology must include, in addition to actual entities,
the \emph{possibility structures} those entities instantiate. These
are not additional entities over and above the physical world; they are
structural features of the physical world--features that are encoded
in the laws of nature and that have measurable physical consequences.

This has implications for the debate between structural realism and
traditional object-oriented realism. The counterfactual reading supports
structural realism's claim that structure is fundamental, but it extends
the notion of structure to include modal structure. The world is not just
a network of actual causal relations; it is a network of \emph{possible}
causal relations, and the information content of physical events is
constituted by this modal network.

\section{Conclusion}
\label{sec:conclusion}

This paper has argued for a specific position in the philosophy of
information: \emph{information as counterfactual structure}. The central
claim is that the informational content of a physical event is not a
property of the event itself but of the possibility space from which
the event was selected. Information exists only where alternatives exist,
and its magnitude is a function of the structure of those alternatives.

Three independently motivated cases support this claim. Landauer's
principle is best understood as the thermodynamic cost of collapsing a
two-valued possibility space to a one-valued space, not as the destruction
of a physical substance. The no-cloning theorem reflects the impossibility
of duplicating quantum possibility spaces, not merely the impossibility of
copying individual quantum states. The black-hole information paradox is
the catastrophic merger of distinct possibility spaces into an
indistinguishable final state, which is why it threatens the foundations
of physics.

The philosophical framework that emerges--modal-structural realism--is a
modest extension of ontic structural realism that takes possibility spaces
seriously as objective features of the physical world. It avoids the
excesses of information physicalism (information is not a substance) and
information abstractionism (information has real physical consequences),
while going beyond standard structural realism by insisting that the
relevant structure is modal as well as actual.

Several questions remain open. A precise account of how possibility spaces
are individuated in theories without a fixed background spacetime--general
relativity, quantum gravity--is needed. The relation between the objective
possibility spaces of the counterfactual framework and the subjective
probability assignments of Bayesian agents requires careful development.
And the implications of modal-structural realism for the interpretation of
quantum mechanics, beyond the sketch offered in \cref{sec:implications},
deserve a dedicated treatment.

What the present analysis does establish is that the question ``Is
information physical?'' is the wrong question. The better question is:
``What role do physically realizable possibilities play in the ontology
of the natural world?'' The answer this paper defends is: a
foundational one.

\printbibliography

\end{document}